\documentclass{llncs}
\usepackage{amsmath,amssymb,mathrsfs}
\usepackage{color}
\usepackage{graphicx,subfigure}
\usepackage{indentfirst}
\usepackage{cite}
\usepackage{algorithm}
\usepackage{algorithmic}
\usepackage{multirow}
\usepackage{xcolor}

\newcommand{\be}{\begin{equation}}
\newcommand{\ee}{\end{equation}}

\renewcommand{\appendix}{\par
\if@chapter@pp
\setcounter{chapter}{0}%
\setcounter{section}{0}%
\gdef\@chapapp{\appendixname}%
\gdef\thechapter{\@Alph\c@chapter}
\else
\setcounter{section}{0}%
\setcounter{subsection}{0}%
\gdef\thesection{Appendix \Alph{section}}
\fi
}

\spnewtheorem{subcase}{Subcase}[case]{\it}{}
\numberwithin{subcase}{case}
\begin{document}
\title{Combinations of Some Shop Scheduling Problems and the Shortest Path Problem: Complexity and Approximation Algorithms}

\author{Kameng Nip, Zhenbo Wang\thanks{Corresponding author. Department of Mathematical Sciences, Tsinghua University, Beijing, 100084, China. Email: zwang@math.tsinghua.edu.cn} \and Wenxun Xing}

\institute{Department of Mathematical Sciences, Tsinghua University, Beijing, 100084, China}
\maketitle
\begin{abstract}
We consider several combinatorial optimization problems which combine the classic shop scheduling problems, namely open shop scheduling or job shop scheduling, and the shortest path problem. The objective of the obtained problem is to select a subset of jobs that forms a feasible solution of the shortest path problem, and to execute the selected jobs on the open (or job) shop machines to minimize the makespan. We show that these problems are $\mathrm{NP}$-hard even if the number of machines is two, and cannot be approximated within a factor less than 2 if the number of machines is an input unless $\mathrm{P}=\mathrm{NP}$. We present several approximation algorithms for these combination problems.

\end{abstract}

\section{Introduction}
Combinatorial optimization has been developed for more than fifty years, and it involves many active subfields, e.g. network flows, scheduling, bin packing, etc. Usually these subfields were arisen by different applications or theoretical interests, and separately developed. The advancement of science and technology makes it possible to integrate manufacturing, service and management. At the same time, the decision-makers always need to deal with problems incurred by more than one combinatorial optimization problems.

Wang and Cui \cite{WC12} introduced a problem combining two classic combinatorial optimization problems, namely parallel machine scheduling and the vertex cover problem. The combination problem is to select a subset of jobs that forms a vertex cover and to schedule it on some identical parallel machines such that the makespan is minimized. They proposed a $(3 - \frac{2}{m+1})$-approximation algorithm. This work also inspired the study of the combination of different combinatorial optimization problems.

Flow shop, open shop and job shop are three basic models of multi-stage scheduling problems. Nip and Wang \cite{NW13} studied a combination problem that combines two-machine flow shop scheduling and the shortest path problem. They argued that this problem is $\mathrm{NP}$-hard, and proposed two approximation algorithms with worst-case ratio $2$ and $\frac{3}{2}$ respectively. Recently they extended the results to the case that the number of flow shop machines is arbitrary \cite{NW13_2}. One motivation of this problem is manufacturing rail racks. We need to build a railway between two cities. How should we choose a feasible path in a map, such that the corresponding rail tracks (jobs) can be manufactured on some flow shop machine as early as possible? It is convincing to change the flow shop environment into the other two well-known shop environments, i.e. open shop and job shop, as they also apply widely in the real world. This is the core motivation for this current work. In this paper, we mainly study two types of problems: the combination of open shop scheduling and the shortest path problem, and the combination of job shop scheduling and the shortest path problem.

The contributions of this paper are described as follows:
(1) we argue that these combination problems are $\mathrm{NP}$-hard even if the number of machines is two, and if the number of machines is an input, these problems cannot be approximated within a factor of $2$ unless $\mathrm{P}=\mathrm{NP}$; (2) we present several approximation algorithms with performance ratio summarized as follows in which $\epsilon > 0$ is any constant and $\mu$ is the maximum number of operations per job in job shop scheduling.
\begin{table}[h]\label{tab1}
\begin{center}
\begin{tabular}{|c|c|c|}
                  \hline
                   Number of Machines & Open Shop & Job Shop\\
                  \hline
                  2 & FPTAS & $\frac{3}{2}+\epsilon$* \\
                  $m$ (fixed) & PTAS** & $O\left(\frac{\log^2(m\mu)}{\log{\log(m\mu)}}\right)$\\
                  $m$ (input) &  $m$ &  $m$\\
                  \hline
\end{tabular}
\caption{Performance of our algorithms}
\end{center}
* Assume that each job has at most $2$ operations.\\
** A $(2+\epsilon)$-approximation algorithm is also proposed.
\end{table}

The rest of the paper is organized as follows. In Section \ref{sec_pre}, we give a formal definition of the combination problems stated above, and briefly review some related problems and algorithms will be used subsequently. In Section \ref{sec_com}, we study the computational complexity of these combination problems and give an inapproximability result when the number of machines is an input. Section \ref{sec_approx} provides several approximation algorithms for these problems. Some concluding remarks are provided in Section \ref{sec_end}.

\section{Preliminaries}\label{sec_pre}
\subsection{Problem Description}\label{sec_pd}
We first define the combination problems considered in this paper.

\begin{definition}
Given a directed graph $G = (V, A)$ with two distinguished vertices $s, t \in V$, and $m$ machines. Each arc $a_j \in A$ corresponds with a job $J_j\in J$. Each job $J_j$ has several operations $O_{1j}$, $O_{2j}$, $\cdots$, $O_{sj}$ (in the open shop, $s = m$ and the order is arbitrary; in the job shop, the order is given as a chain). The processing times for $J_j$ on machine $M_i$ is $p_{ij}$. The $Om|\mathrm{shortest}~\mathrm{path}|C_{max}$ ($Jm|\mathrm{shortest}~\mathrm{path}|C_{max}$) problem is to find a $s-t$ directed path $P$ of $G$, and to schedule the jobs of $J_P$ on the open (job) shop machines to yield the minimum makespan over all $P$, where $J_P$ denotes the set of jobs corresponding to the arcs in $P$.\label{d_comb}
\end{definition}

We denote the number of jobs (arcs) as $n$, i.e. $|A| = |J| = n$. Denote by $\{M_1,M_2,\cdots, M_m\}$ the $m$ machines, and let $\mu_{ij}$ be the times of $J_j$ needed to be processed on $M_i$. Notice $\mu_{ij} = 1$ in the open shop.

It is not difficult to see that the open (job) shop scheduling problem and the classic shortest path problem are special cases of our problems, and hence we say the considered problems are the combinations of the scheduling problems and the shortest path problem. We will show that the combination problems appear different aspects in computational complexity and algorithm design from the shop scheduling problems or the shortest path problem.

In this paper, we will use the results of some optimization problems that
have a similar structure with the classic shortest path problem. We introduce the generalized shortest path problem defined in \cite{NW13}, and extend it to $K$ weights.

\begin{definition}
Given a directed graph $G = (V, A, w^1, \cdots, w^K)$ and two distinguished vertices $s, t\in V$ with $|A| = n$. Each arc $a_j\in A, j = 1,\cdots,n$ is associated with $K$ weights $w^1_j, \cdots, w^K_j$, and we define vector $w^k = (w^k_1, w^k_2, \cdots, w^k_n)$ for $k=1, 2, \cdots, K$. The goal of our shortest path problem $SP(G, s, t, f)$ is to find a $s-t$ directed path $P$ that minimizes $f(w^1, w^2, \cdots, w^K; x)$, in which $f$ is a given objective function and $x \in \{0, 1\}^n$ contains the decision variables such that $x_j = 1$ if and only if $a_j\in P$.
\label{d_sp}
\end{definition}

For simplicity of notation, we denote $SP$ instead of $SP(G, s, t, f)$ in the rest of the paper. Notice $SP$ is a generalization of various shortest path problems. For example, if we set $K = 1$ and $f(w^1, x) = w^1\cdot x$, where $\cdot$ is the dot product, it is the classic shortest path problem. If $f(w^1\cdot x, w^2\cdot x, \cdots, w^K\cdot x, x) = \max\{w^1\cdot x, w^2\cdot x, \cdots, w^K\cdot x\}$, it is the min-max shortest path problem \cite{ABV06}.

\subsection{Review of Open Shop and Job Shop Scheduling}\label{sec_f2}
Gonzalez and Sahni \cite{Gonzalez1976} first gave a linear time optimal algorithm for $O2||C_{max}$. They also proved that $Om||C_{max}$ is $\mathrm{NP}$-hard for $m\geq 3$, however whether it is strongly $\mathrm{NP}$-hard is still an outstanding open problem. A feasible shop schedule is called dense when any machine is idle if and only if there is no job that could be processed on it. R{\'a}csm{\'a}ny (see B{\' a}r{\' a}ny and Fiala \cite{bt82}) observed that for any dense schedule, the makespan is at most twice of the optimal solution, that leads to a greedy algorithm. Sevastianov and Woeginger \cite{Sevastianov1998} presented a PTAS for fixed $m$, which is obtained by dividing jobs into large jobs and small jobs. Their algorithm first  optimally schedules the large jobs, then fills the operations of the small jobs into the `gaps'. In this paper, we will use these algorithms, and refer to the GS algorithm, R{\'a}csm{\'a}ny algorithm and the SW algorithm respectively. We present the main results of these algorithms as follows.

\begin{theorem}[\cite{Gonzalez1976}]
The GS algorithm returns an optimal schedule for $O2||C_{max}$ in linear time such that
$
C_{max} = \max\left\{\max_{J_j\in J}(p_{1j} + p_{2j}), \sum_{J_j\in J}p_{1j}, \sum_{J_j\in J}p_{2j}\right\}.
$
\label{th_o_gs}
\end{theorem}
\begin{theorem}[\cite{bt82,Shmoys1994}]
R{\'a}csm{\'a}ny algorithm returns a 2-approximation algorithm for $Om||C_{max}$ such that
$
C_{max} \leq \sum_{J_j\in J}p_{lj} + \sum_{i = 1}^m p_{ik} \leq 2C^*_{max},
$
where $J_k$ is the last completed job and processed on $M_l$.\label{th_o_racsmany}
\end{theorem}

\begin{theorem}[\cite{Sevastianov1998}]
The SW algorithm is a PTAS for $Om||C_{max}$.\label{th_o_sw}
\end{theorem}
For job shop schedule with an unlimited number of jobs, few polynomially solvable cases are known. One is $J2|op\leq 2|C_{max}$, which can be solved by Jackson's rule \cite{Jackson56} that is an extension of Johnson's rule for $F2||C_{max}$ \cite{Johnson54}, where $op\leq 2$ means there are at most $2$ operations per job. The idea is to divide the jobs into two sets according to the processing order of the jobs, and implement Johnson's rule for each job set, then combine the schedules. In fact, a slightly change may lead to $\mathrm{NP}$-hard problems. For instance, $J2|op\leq 3|C_{max}$ and $J3|op\leq 2|C_{max}$ are $\mathrm{NP}$-hard \cite{Lenstra1977}, $J2|p_{ij}\in\{1,2\}|C_{max}$ and $J3|p_{ij}=1|C_{max}$ are strongly $\mathrm{NP}$-hard \cite{Lenstra1979}. For the general case $J||C_{max}$, Shmoys, Stein and Wein \cite{Shmoys1994} constructed a randomized approximation algorithm with worst-case ratio $O\left(\frac{\log^2(m\mu)}{\log{\log(m\mu)}}\right)$, where $\mu$ is the maximum number of operations per job. Schmidt, Siegel and Srinivasan \cite{Schmidt1995} obtained a deterministic algorithm with the same bound by derandomizing. We refer to it as the SSW-SSS algorithm. Moreover, for fixed $m$, the best known approximation algorithm is also proposed in \cite{Shmoys1994} with an approximation factor $2 + \epsilon$, where $\epsilon > 0$ is an arbitrary constant. If $\mu$ is a constant, the problem is denoted as $Jm|op\leq \mu|C_{max}$ that admits a PTAS \cite{Jansen2003}. We list the main results mentioned above as follows.

\begin{theorem}[\cite{Jackson56}]
Jackson's rule solves $J2|op\leq 2|C_{max}$ in $O(n\log n)$ time. \label{th_j_jackson}
\end{theorem}

\begin{theorem}[\cite{Shmoys1994,Schmidt1995}]
The SSW-SSS algorithm solves $Jm||C_{max}$ in polynomial time, and return a schedule with makespan
$$O\left(\frac{\log^2(m\mu)}{\log{\log(m\mu)}} \left(\max_{i\in\{1,\cdots,m\}}{\sum_{J_j\in J}\mu_{ij}p_{ij}}+ \max_{J_j\in J}{\sum^{m}_{i=1}\mu_{ij}p_{ij}}\right)\right).$$\label{th_SSW_SSS}
\end{theorem}

Furthermore, a well-known inapproximability result is that $O||C_{max}$, $F||C_{max}$ and $J||C_{max}$ cannot be approximated within $\frac{5}{4}$ unless $\mathrm{P}=\mathrm{NP}$ \cite{Williamson97}. Recently, Mastrolilli and Svensson \cite{Mastrolilli:2011:HAF:2027216.2027218} showed that $J||C_{max}$ cannot be approximated within $O(\log(m\mu)^{1-\epsilon})$ for $\epsilon > 0$ based on a stronger assumption than $\mathrm{P}\neq\mathrm{NP}$.

To conclude this subsection, we list some trivial bounds for a dense shop schedule. Denote by $C_{max}$ the makespan of an arbitrary dense shop schedule with job set $J$, and we have

\be
C_{max} \geq \max_{i\in\{1,\cdots,m\}}\left\{\sum_{J_j\in J}\mu_{ij}p_{ij}\right\},\label{eq_max}
\ee
and
\be
C_{max} \leq \sum_{J_j\in J}\sum^m_{i=1}\mu_{ij}p_{ij}.\label{eq_min}
\ee
For each job, we have
\be
C_{max} \geq \sum^m_{i=1}\mu_{ij}p_{ij},\qquad \forall J_j \in J.\label{eq_job}
\ee

\subsection{Review of Shortest Path Problems}\label{sec_sp}
It is well-known that Dijkstra algorithm solves the classic shortest path problem with nonnegative edge weights in $O(|V|^2)$ time \cite{DIJ59}.
We have mentioned the min-max shortest path problem, that is $\mathrm{NP}$-hard even for $K = 2$, and Aissi, Bazgan and Vanderpooten proposed a FPTAS if $K$ is a fixed number \cite{ABV06}. We refer to their algorithm as the ABV algorithm, which has the following result.

\begin{theorem}[\cite{ABV06}]\label{th_minmax}
Given $\epsilon > 0$, in a directed graph with $K$ nonnegative weights on each arc, where $K$ is a fixed number. The ABV algorithm finds a path $P$ between two specific vertices satisfying $\max_{i\in \{1, 2, \cdots, K\}} \left\{\sum_{a_j\in P}w^i_j\right\} \leq (1+\epsilon) \max_{i\in \{1, 2, \cdots, K\}}\left\{\sum_{a_j\in P'}w^i_j\right\}$
for any path $P'$ between the two specified vertices, and the running time is $O(|A||V|^{K + 1}/\epsilon^K)$.
\end{theorem}

In this paper, sometimes we need to find the min-mix shortest path among all the paths visiting some specified arcs if such a path exists. We propose a modified ABV algorithm for this problem in \ref{app_modABV}.

\section{Computational Complexity}\label{sec_com}
First, notice that $Om||C_{max}$ and $Jm||C_{max}$ are special cases of the corresponding combination problems, thus the combination problem is not easier than its component optimization problems. On the other hand, we know $O2||C_{max}$ and $J2|op\leq 2|C_{max}$ are polynomially solvable, but we can simply verify that the corresponding combination problems, say $O2|\mathrm{shortest}~\mathrm{path}|C_{max}$ and $J2|op\leq 2, \mathrm{shortest}~\mathrm{path}|C_{max}$, are $\mathrm{NP}$-hard by adopting the same reduction proposed in \cite{NW13} for the $\mathrm{NP}$-hardness of $F2|\mathrm{shortest}~\mathrm{path}|C_{max}$. We summarize the results as Theorem \ref{comp_m}.

\begin{theorem}
Even if $m = 2$, $Jm|\mathrm{shortest}~\mathrm{path}|C_{max}$ is strongly $\mathrm{NP}$-hard and $Om|\mathrm{shortest}~\mathrm{path}|C_{max}$ is $\mathrm{NP}$-hard. $J2|op\leq 2, \mathrm{shortest}~\mathrm{path}|C_{max}$ is $\mathrm{NP}$-hard. \label{comp_m}
\end{theorem}

Now we consider the case where the number of machines $m$ is part of the input. Williamson et al. showed that it is $\mathrm{NP}$-hard to approximate $O||C_{max}$, $F||C_{max}$ or $J||C_{max}$ within a factor less than $\frac{5}{4}$ by a reduction from the restricted versions of 3-SAT \cite{Williamson97}. They also showed that deciding if there is a scheduling of length at most 3 is in $\mathbf{P}$.
We show that for these problems combining with shortest path problem, deciding if there is a scheduling of length at most 1 is still $\mathrm{NP}$-hard. Our proof is established by constructing a reduction from 3-Dimensional Matching (3DM) that is $\mathrm{NP}$-complete \cite{GJ79}.

{\sc 3-Dimensional Matching}:\\
{\bf Instance:} Sets $A = \{a_1, \cdots, a_n\}$, $B = \{b_1, \cdots, b_n\}$, $C = \{c_1, \cdots, c_n\}$, and a family $F = \{T_1, \cdots, T_m\}$ of triples with $|T_k\cap A| = |T_k\cap B| = |T_k\cap C| = 1$ for $k = 1, \cdots, m$. Assume that $m\geq n$ without loss of generality.\\
{\bf Question:} Does $F$ contains a matching, i.e. a subfamily $F'$ for which $|F'| = n$ and $\cup_{T_k\in F'}T_k = A \cup B \cup C$?

\begin{theorem}
For $O|\mathrm{shortest}~\mathrm{path}|C_{max}$, deciding if there is a scheduling of length at most $1$ is $\mathrm{NP}$-hard. \label{th_inapp}
\end{theorem}
\begin{proof}
Given an instance of 3DM, we construct an instance of $O|\mathrm{shortest}~\mathrm{path}|C_{max}$ with $2n + m$ machines. For $i = 1, \cdots, n$, machines $M_i$, $M_{n + i}$ and $M_{2n + i}$ correspond to $a_i$, $b_i$ and $c_i$ respectively, and the remained $m - n$ machines denoted by $M_{3n + 1},M_{3n + 2}\cdots, M_{2n+m}$ are `dummy' machines. The graph has $3m + 1$ vertices, denoted by $v_{1,a}, v_{1,b}, v_{1,c}, v_{2,a}, v_{2,b}, v_{2,c}, \cdots, v_{m, a}, v_{m, b}, v_{m, c}, v_{m+1,a}$. For $k = 1, \cdots, m$, there are arcs $(v_{k,a}, v_{k,b})$, $(v_{k,b}, v_{k,c})$ and$(v_{k,c}, v_{k + 1,a})$ corresponding to jobs $J^a_k, J^b_k, J^c_k$. Let $p(J^j_k)$, $j=a,b,c$ and $k = 1, \cdots, m$, be a $(2n + m)$-dimensional vector whose $i$-th component corresponds to the processing time of $J^j_k$ on $M_i$. Let $\mathbf{e}_{i}$ be a $(2n + m)$-dimensional vector with 1 for the $i$-th component and $0$s for the others. Now we can define the processing times of the jobs: $p(J^a_k) = \mathbf{e}_i$ if $a_i\in T_k$; $p(J^b_k) = \mathbf{e}_{n+i}$ if $b_i\in T_k$; $p(J^c_k) =\mathbf{e}_{2n+i}$ if $c_i\in T_k$. Selecting jobs $J^a_k, J^b_k, J^c_k$ implies that $T_k$ is in the matching. Moreover for $k = 1, \cdots, m$, there are $m - n$ parallel arcs from $v_{k,a}$ towards $v_{k+1, a}$, corresponding to jobs $J^1_k, J^2_k, \cdots, J^{m-n}_k$ with processing times $\mathbf{e}_{3n + 1}, \mathbf{e}_{3n + 2}, \cdots, \mathbf{e}_{2n+m}$ respectively. Selecting those jobs implies that $T_k$ is not in the matching. The objective is to find a path from $v_{1,a}$ to $v_{m+1,a}$ and to schedule the corresponding jobs (arcs) such that the makespan is at most $1$, that completes the reduction. One example is shown in Figure \ref{fig_inapp_2}.

\begin{figure}[ht]
  \centering
  \includegraphics[width=5in]{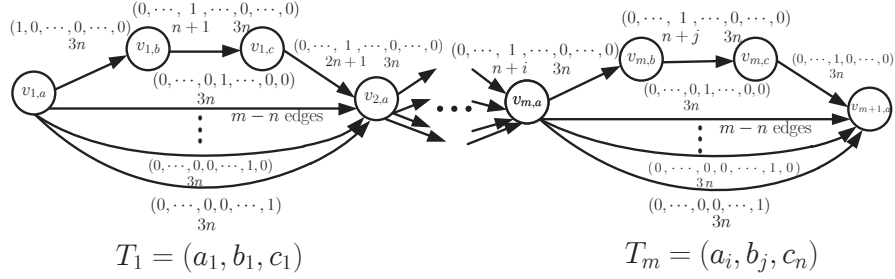}\\
  \caption{An example of the reduction with $T_1 = \{a_1, b_1, c_1\}, T_m = \{a_i, b_j, c_n\}$.}\label{fig_inapp_2}
\end{figure}

It can be verified that a schedule with makespan at most 1 if and only if there is a matching for 3DM, and then the result follows. The details are deferred to the full version.
\qed
\end{proof}

Notice that the reduction in Theorem \ref{th_inapp} is also valid for $F|\mathrm{shortest}~\mathrm{path}|C_{max}$ and $J|\mathrm{shortest}~\mathrm{path}|C_{max}$, since each job in the reduction has only one nonzero processing time. Therefore we have the following result.
\begin{corollary}
$O||C_{max}$, $F||C_{max}$ and $J||C_{max}$ do not admit an approximation algorithm with worst-case ratio less than $2$, unless $\mathrm{P}=\mathrm{NP}$.
\end{corollary}

To our knowledge, the best known inapproximability results based on $\mathrm{P}\neq\mathrm{NP}$ for $F||C_{max}$, $O||C_{max}$ and $F||C_{max}$ are still $\frac{5}{4}$. The corollary implies that the combination problems of the three shop scheduling problems and the shortest problem have stronger inapproximability results.

\section{Approximation Algorithms}\label{sec_approx}
\subsection{An Intuitive Algorithms for Arbitrary $m$}\label{sec_alg_nat}
An intuitive algorithm was proposed for $F2|\mathrm{shortest}~\mathrm{path}|C_{max}$ in \cite{NW13}. The idea is to find the classic shortest path by setting the weight of an arc be the sum of processing times of its corresponding job, and then schedule the returned jobs by Johnson's rule. This simple idea can be extended to all the combination problems we considered, even if the number of machines is an input.

\begin{algorithm}[htb]
\caption{The SD algorithm for {\small$O|\mathrm{shortest}~\mathrm{path}|C_{max}$ ($J|\mathrm{shortest}~\mathrm{path}|C_{max}$)}}
\label{alg_1}
\begin{algorithmic}[1]
\STATE Find the shortest path in $G$ with weights $w^1_j := \sum^m_{i = 1}\mu_{ij}p_{ij}$ by Dijkstra algorithm. For the returned path $P$, construct the job set $J_P$.

\STATE Obtain a dense schedule for the jobs of $J_P$ by an arbitrary open (job) shop scheduling algorithm. Let $\sigma$  be the returned job schedule and $C_{max}$ the returned makespan, and denote the job set $J_P$ by $S$.
\RETURN $S$, $\sigma$ \AND $C_{max}$.
\end{algorithmic}
\end{algorithm}

It is easy to show that Algorithm \ref{alg_1} is a $m$-approximation algorithm,
by the bounds (\ref{eq_max}), (\ref{eq_min}), (\ref{eq_job}) and the fact that the returned path is the shortest path with respect to the single weight of each arc.



\subsection{A Unified Algorithms for Fixed $m$}\label{sec_alg_im}
In \cite{NW13}, a $\frac{3}{2}$-approximation algorithm was proposed for $F2|\mathrm{shortest}~\mathrm{path}|C_{max}$. The idea is to iteratively find a feasible path by the ABV algorithm with two weights for each arcs and schedule the corresponding jobs by Johnson's rule, and then adaptively modified the weights of arcs. We generalize this idea to solve the combination problems considered in this paper. We first propose a unified framework which denoted as UAR($Alg$, $\rho$, $m$), where $Alg$ is a polynomial time algorithm used for shop scheduling, $\rho$ is a control parameter to decide the termination rule of the iterations and the jobs to be modified, and $m$ is the number of machines.

\begin{algorithm}[htb]
\caption{Algorithm UAR($Alg$, $\rho$, $m$)}
\label{alg_uar}
\begin{algorithmic}[1]
\STATE Initially,$(w^1_j, w^2_j, \cdots, w^m_j) :=(\mu_{1j}p_{1j}, \mu_{2j}p_{2j}, \cdots, \mu_{mj}p_{mj})$, for $a_j\in A$ corresponding to $J_j$.

\STATE Given $\epsilon >0$, implement the ABV algorithm to obtain a feasible path $P$ to $SP$, and construct the corresponding job set as $J_P$.

\STATE Schedule the jobs of $J_P$ by the algorithm $Alg$, denote the returned makespan as $C'_{max}$, and the job schedule as $\sigma'$.

\STATE $S : = J_P$, $\sigma: = \sigma'$, $C_{max}:=C'_{max}$, $D:=\emptyset$, $M : = (1+\epsilon)\sum_{J_j\in J}\sum_{i = 1}^{m}\mu_{ij}p_{ij} + 1$.

\WHILE{$J_P \cap D = \emptyset $ \AND there exists $J_j$ in $J_P$ satisfies $\sum_{i=1}^m\mu_{ij}p_{ij} \geq \rho C'_{max}$}
    \FOR {all jobs satisfy $\sum_{i=1}^m\mu_{ij}p_{ij} \geq \rho C'_{max}$ in $J\backslash D$}
        \STATE $(w^1_j, w^2_j, \cdots, w^m_j) := (M, M, \cdots, M)$, $D := D\cup \{J_j\}$.
    \ENDFOR
    \STATE Implement the ABV algorithm to obtain a feasible path $P$ to $SP$, and construct the corresponding job set as $J_P$.

    \STATE Schedule the jobs of $J_P$ by the algorithm $Alg$, denote the returned makespan as $C'_{max}$, and the job schedule as $\sigma'$.

    \IF{$C'_{max} < C_{max}$}
        \STATE $S : = J_P$, $\sigma: = \sigma'$, $C_{max}:=C'_{max}$.
    \ENDIF
\ENDWHILE
\RETURN $S$, $\sigma$ \AND $C_{max}$.
\end{algorithmic}
\end{algorithm}

By setting the appropriate scheduling algorithms and control parameters, we can derive algorithms for different combination problems. Notice that at most $n$ jobs are modified in the UAR($Alg$, $\rho$, $m$) algorithm, therefore the iterations execute at most $n$ times. Since the scheduling algorithms for shop scheduling and the ABV algorithms are all polynomial time algorithms (for fixed $m$ and $\epsilon$), we claim that the following algorithms based on UAR($Alg$, $\rho$, $m$) are polynomial-time algorithms. We present the algorithms and their performance as follows and the detailed proofs are given in \ref{app_proof}.

We first apply the UAR($Alg$, $\rho$, $m$) algorithm to $O2|\mathrm{shortest}~\mathrm{path}|C_{max}$ by setting $Alg$ be the GS algorithm and $\rho = 1$, and refer to this algorithm as the GAR algorithm.

\begin{algorithm}[htb]
\caption{The GAR algorithm for $O2|\mathrm{shortest}~\mathrm{path}|C_{max}$}
\label{alg_gar}
\begin{algorithmic}[1]
\STATE Set $m = 2$, $Alg$ be the GS algorithm for $O2||C_{max}$ and $\rho$ = 1.
\STATE Solve the problem by implementing UAR($Alg$, $\rho$, $m$).
\end{algorithmic}
\end{algorithm}

\begin{theorem}
The GAR algorithm is a FPTAS for $O2|\mathrm{shortest}~\mathrm{path}|C_{max}$.  \label{th_GAR}
\end{theorem}

For $Om|\mathrm{shortest}~\mathrm{path}|C_{max}$ where $m$ is fixed, based on UAR($Alg$, $\rho$, $m$) and R{\'a}csm{\'a}ny algorithm, we obtain the following algorithm, referred to the RAR algorithm by considering an appropriate $\rho$.

\begin{algorithm}[htb]
\caption{The RAR algorithm for $Om|\mathrm{shortest}~\mathrm{path}|C_{max}$}
\label{alg_rar}
\begin{algorithmic}[1]
\STATE Set $Alg$ be R{\'a}csm{\'a}ny algorithm for $Om||C_{max}$ and $\rho = \frac{1}{2}$.
\STATE Solve the problem by implementing UAR($Alg$, $\rho$, $m$).
\end{algorithmic}
\end{algorithm}
\begin{theorem}
Given $\epsilon > 0$, the RAR algorithm is a $(2+\epsilon)$-approximation algorithm for $Om|\mathrm{shortest}~\mathrm{path}|C_{max}$.  \label{th_RAR}
\end{theorem}

The framework also can be applied to the combination problem of job shop scheduling and the shortest path problem. For the combination of $J2|op\leq 2|C_{max}$ and the shortest path problem, we obtain a $(\frac{3}{2}+ \epsilon)$-approximation algorithm by implementing Jackson's rule and setting $\rho = \frac{2}{3}$ in the UAR($Alg$, $\rho$, $m$) algorithm. We refer to this algorithm as the JJAR algorithm, and describe it in Algorithm \ref{alg_jjar}. Remind that all $\mu_{ij} = 1$ in $J2|op\leq 2|C_{max}$.

\begin{algorithm}[htb]
\caption{The JJAR algorithm for $J2|op\leq 2, \mathrm{shortest}~\mathrm{path}|C_{max}$}
\label{alg_jjar}
\begin{algorithmic}[1]
\STATE Set $m = 2$, $Alg$ be Jackson's rule for $J2|op\leq 2|C_{max}$ and $\rho = \frac{2}{3}$.
\STATE Solve the problem by implementing UAR($Alg$, $\rho$, $m$).
\end{algorithmic}
\end{algorithm}

Before studying the worst-case performance of the JJAR algorithm, we establish the following lemma. Let $(1\rightarrow2)$ ($(2\rightarrow1)$) indicate the order that a job needs to be processed on $M_1$ ($M_2$) first and then on $M_2$ ($M_1$).

\begin{lemma}
For $J2|op \leq 2|C_{max}$, let $C^J_{max}$ be the makespan returned by Jackson's rule. Suppose we change the processing order of all jobs to be $(1\rightarrow2)$ $((2\rightarrow1))$, and the processing times keep unchanged. Then schedule the jobs by Johnson's rule for $F2||C_{max}$, and denote the makespan as $C^1_{max}$ $(C^2_{max})$. We have $C^J_{max}\leq \max\{C^{1}_{max}, C^2_{max}\}$. \label{lemma_J2}
\end{lemma}

The proof of lemma \ref{lemma_J2} is also given in \ref{app_proof}.
\begin{theorem}
Given $\epsilon > 0$, the JJAR algorithm is a $(\frac{3}{2}+\epsilon)$-approximation algorithm for $J2|op\leq 2, \mathrm{shortest}~\mathrm{path}|C_{max}$.  \label{th_jjar}
\end{theorem}

Finally, we study the general case $Jm|\mathrm{shortest}~\mathrm{path}|C_{max}$, where $m$ is fixed. By theorem \ref{th_SSW_SSS}, we know that there exists $\alpha > 0$, such that the SSW-SSS algorithm returns a schedule satisfies
\be
C'_{max} \leq\alpha\frac{\log^2(m\mu)}{\log{\log(m\mu)}}\left(\max_{i\in\{1,\cdots,m\}}{\sum_{J_j\in J'}\mu_{ij}p_{ij}}+ \max_{j\in J'}{\sum^{m}_{i=1}\mu_{ij}p_{ij}}\right).\label{eq:sjar_1}
\ee
The factor $\alpha$ is decided by choosing the probability of the randomized steps and the subsequent operations in the SSW-SSS algorithm \cite{Shmoys1994,Schmidt1995}, and its value can be obtained by complicated calculation. Assume we determine such value of $\alpha$. We can design an approximation algorithm with worst-case ratio $O\left(\frac{\log^2(m\mu)}{\log{\log(m\mu)}}\right)$ for $Jm|\mathrm{shortest}~\mathrm{path}|C_{max}$. We refer to this algorithm as the SAR algorithm, and describe it in Algorithm \ref{alg_sar}.

\begin{algorithm}[htb]
\caption{The SAR for $Jm|\mathrm{shortest}~\mathrm{path}|C_{max}$}
\label{alg_sar}
\begin{algorithmic}[1]
\STATE Set $Alg$ be the SSW-SSS algorithm for $Jm||C_{max}$ and $\rho = \frac{\log{\log(m\mu)}}{2\alpha\log^2(m\mu)}$.
\STATE Solve the problem by implementing UAR($Alg$, $\rho$, $m$).
\end{algorithmic}
\end{algorithm}

\begin{theorem}
The SAR algorithm is an $O\left(\frac{\log^2(m\mu)}{\log{\log(m\mu)}}\right)$-approximation algorithm for $Jm|\mathrm{shortest}~\mathrm{path}|C_{max}$.  \label{th_SAR}
\end{theorem}

However, we remind that the SAR algorithm relies on the assumption, that we can determine the constant $\alpha$ for the SSW-SSS algorithm. We can calculate it by following the details of the SSW-SSS algorithm, and in fact we can choose $\alpha$ large enough to guarantee the performance ratio of our algorithm.

\subsection{A PTAS for $Om|\mathrm{shortest}~\mathrm{path}|C_{max}$}\label{sec_alg_1+e_o}
In the previous subsection, we introduced a $(2+\epsilon)$-approximation algorithm for $Om|\mathrm{shortest}~\mathrm{path}|C_{max}$ based on the UAR($Alg$, $\rho$, $m$) algorithm. By a different approach, we propose a $(1+\epsilon)$-approximation algorithm for any $\epsilon>0$, i.e. a PTAS. We also iteratively find feasible solutions, but guarantee that one of the returned solutions has the same first $N$-th largest jobs with an optimal solution where $N$ is a given constant. Precisely speaking, we say job $J_j$ is larger than job $J_k$ if $\max_{i\in \{1, \cdots, m\}}p_{ij}>\max_{i\in \{1, \cdots, m\}}p_{ik}$. To do this, we enumerate all size $N$ subsets $J^N$ of $J$, and then iteratively modify the weights of the graph such that the jobs larger than any job in $J^N$ will not be chosen. Then find a feasible solution which contains all the jobs in $J^N$ corresponding to the modified graph, i.e., the corresponding path is constrained to visit all the arcs corresponding to $J^N$ if such a path exists.

To find a feasible solution in each iteration, we adopt the modified ABV algorithm (see \ref{app_modABV}) to obtain a near optimal min-max shortest path among all the paths visiting the arcs corresponding to $J^N$ if such a path exists. Then we schedule the selected jobs by the PTAS for $Om||C_{max}$ \cite{Sevastianov1998} which is denoted as the SW algorithm. We refer to our algorithm as the SAE algorithm, and describe it in Algorithm \ref{alg_sae}.

\begin{algorithm}[htb]
\caption{The SAE algorithm for $Om|\mathrm{shortest}~\mathrm{path}|C_{max}$}
\label{alg_sae}
\begin{algorithmic}[1]
\STATE Given $0< \epsilon <1$, set $N = m\left(\frac{m(3+\epsilon)}{\epsilon}\right)^{2^{\frac{m(3+\epsilon)}{\epsilon}}}$.
\STATE Set $D:=\emptyset$, $M : = (1+\frac{\epsilon}{3})\sum_{J_j\in J}\sum_{i = 1}^{m}p_{ij} + 1$.
\STATE Initially, $(w^1_j, w^2_j, \cdots, w^m_j) := (p_{1j}, p_{2j}, \cdots, p_{mj})$, for $a_j\in A$ corresponding to $J_j$.
\FOR {all $J^N \subset J$, with $|J^N| = N$}\label{alg_sae_lfor}
\STATE $(w^1_j, w^2_j, \cdots, w^m_j) := (p_{1j}, p_{2j}, \cdots, p_{mj})$, $D:=\emptyset$.
\STATE For jobs $J_k \in J\setminus J^N$ with $\max_{i\in \{1, \cdots, m\}}p_{ik} > \min_{J_j\in J^N}\max_{i\in \{1, \cdots, m\}}p_{ij}$, set $(w^1_k, w^2_k, \cdots, w^m_k) := (M, M, \cdots, M)$, $D:= D\cup \{J_k\}$.
\STATE Implement the modified ABV algorithm to obtain a feasible path $P$ of $SP$ such that the returned path visits all the arcs corresponding to $J^N$ if such a path exists. Construct the corresponding job set as $J_P$.
\STATE Schedule the jobs of $J_P$ by the SW algorithm.
\IF{$C'_{max} < C_{max}$}
        \STATE $S : = J_P$, $\sigma: = \sigma'$, $C_{max}:=C'_{max}$.
    \ENDIF
\ENDFOR\label{alg_sae_lendfor}
\RETURN $S$, $\sigma$, $C_{max}$.
\end{algorithmic}
\end{algorithm}

There are ${n \choose N}$ distinct subsets $J^N$, thus the iterations between line \ref{alg_sae_lfor} - line \ref{alg_sae_lendfor} run at most $O(n^N)$ times, that is a polynomial of $n$ since $N$ is a constant when $m$ and $\epsilon$ are fixed. Since the modified ABV algorithm is a FPTAS and the SW algorithm is a PTAS, the running time of each iteration is also bounded by the polynomial of $n$ if $m$ and $\epsilon$ are fixed. It suffices to show that the SAE algorithm terminates in polynomial time. The following theorem indicates the SAE algorithm is a PTAS, and detailed proof can be found in \ref{app_SAE_proof}.

\begin{theorem}
The SAE algorithm is a PTAS for $Jm|\mathrm{shortest}~\mathrm{path}|C_{max}$.\label{th_SAE}
\end{theorem}

\section{Conclusions}\label{sec_end}
This paper studies several problems combining two well-known combinatorial optimization problems. We show the hardness of the problems, and present
some approximation algorithms. It is interesting to find approximation algorithms with better worst-case ratios for $J2|op\leq 2, \mathrm{shortest}~\mathrm{path}|C_{max}$ and $Jm|\mathrm{shortest}~\mathrm{path}|C_{max}$. Moreover, it needs further study to close the gap between the 2-inapproximability results and the $m$-approximation algorithms for $O|\mathrm{shortest}~\mathrm{path}|C_{max}$ and $J|\mathrm{shortest}~\mathrm{path}|C_{max}$. One can also consider other interesting combination of combinatorial optimization problems.

\section*{Acknowledgments}
This work has been supported by the Bilateral Scientific Cooperation Project BIL10/10 between Tsinghua University and KU Leuven.

\bibliographystyle{plain}
\bibliographystyle{splncs}

\newpage
\appendix
\section{The Modified ABV Algorithm}\label{app_modABV}
In this appendix, we give a modified version of the ABV algorithm \cite{ABV06}. The algorithm can determine a near optimal min-max shortest path among all the paths which are constrained to visit all arcs in a arc set $A' \subset A$ ($|A'|=N$ is a constant) if such a path exists. We propose a dynamic programming to solve this problem in pseudo-polynomial time.

We index the vertex set $V$ as $\{0, 1, \cdots, t\}$, where $0$ is the starting point and $t$ is the destination. We denote $A' = \{a_1, a_2, \cdots, a_N\}$. For $u,v \in\{ 0,1, \cdots, t\}$, let $S^u_v$ be a set of $(K + N)$-dimensional vectors, where $K$ is the number of weights of each arc. A vector $s(P) \in S^u_v$ with respect to a $0-v$ path $P$ with at most $u$ arcs, satisfies that its first $K$ components are the lengths of $P$ from $0$ to $v$ with respect to different weights respectively, and each of the last $N$ components is associated with an arc in $A'$ such that the $(K + i)$-th component is 1 if $P$ visits $a_i$, and is 0 otherwise.

Notice that each of the first $K$ components of $s(P)$ can be bounded by $W = \max_{k \in \{1, \cdots, K\}}\sum_{j = 1}^{n}{w^k_{j}}$ and the other $N$ components are binary, and thus the size (numbers of distinct vectors) of $S^u_v$ is no more than $W^K2^N = O(W^K)$. Let $s_j^0$ be a $(K + N)$-dimensional vector that is obtained by adding $N$ zeros after $(w_j^1,\cdots,w_j^K)$, and $s_j^1$ is the same with $s_j^0$ except its $(K + j)$-th component is 1. The dynamic programming is described in Algorithm \ref{alg_modABV}.

\begin{algorithm}[htb]
\caption{The Modified ABV Algorithm}
\label{alg_modABV}
\begin{algorithmic}[1]
\STATE $W = \max_{k \in \{1, \cdots, K\}}\sum_{j = 1}^{n}{w^k_{j}}$.
\STATE $S^0_0 : = \{(0, 0, \cdots, 0)\}$, $S^u_v:=\emptyset$ for $u,v \in\{ 0,1, \cdots, t\}$.
\FOR {$u = 1, \cdots, t$}
    \FOR {$v = 1, \cdots, t$}
        \FOR {each $v'$ with $a_j = (v',v) \in A$}
            \FOR {each vector $s(P)=(a^1, \cdots, a^K, a^{K + 1}, \cdots, a^{K + N}) \in S^{u-1}_{v'}$}
            \IF {$\forall i$, $a^i + w^i_j \leq W$}
                \IF {$a_j \in A'$}
                 \STATE $S^u_v:= S^u_v \cup \{s(P)+s_j^1\}$.
                 \ELSE
                 \STATE $S^u_v:= S^u_v \cup \{s(P)+s_j^0\}$.
                \ENDIF
            \ENDIF
            \ENDFOR
        \ENDFOR
    \ENDFOR
\ENDFOR
\RETURN the path $P$ corresponding to a vector $(a^1, \cdots, a^K, 1, \cdots, 1) \in S^t_t$ such that $\max_{i=1, \cdots, K}a^i$ is the minimum among all such vectors if such a vector exists, and otherwise return an empty set.
\end{algorithmic}
\end{algorithm}

It is not difficult to see that Algorithm \ref{alg_modABV} returns an optimal solution in time $O(|V|^2|A|W^K)$, which is a pseudo-polynomial time algorithm. Based on the scaling technique (for example, see\cite{ABV06}), we can use Algorithm \ref{alg_modABV} to derive a FPTAS, such that given $\epsilon >0$, it returns a path with the value at most $(1+\epsilon) \min_P\max_{k=1, \cdots, K}\sum _{a_j\in P}w^k_j$ among all the paths $P$ visiting arc set $A'$ if such a path exists.

\section{The Performance Analysis of Algorithms in Section \ref{sec_alg_im}}\label{app_proof}
We point out that the proofs of the worst-case performance of algorithms based on UAR($Alg$, $\rho$, $m$) are quite similar. We give the detailed proof for the GAR algorithm, and describe the key ideas and main steps for the other results since they can be obtained by analogous arguments.

\begin{proof}[Theorem \ref{th_GAR}]
Let $C_{max}$ and $S$ be the makespan and the job set returned by the GAR algorithm respectively, and $C^*_{max}$ and $J^*$ the value of an optimal solution. We consider two cases.
\setcounter{case}{0}
\begin{case} $J^*\cap D \neq \emptyset$

It implies that there is at least one job in the optimal solution, say $J_j$, such that $p_{1j} + p_{2j} \geq C'_{max}$ holds for a current schedule with makespan $C'_{max}$ during the execution. Notice that the schedule returned by the GAR algorithm is the best one among all current schedules, i.e. $C_{max} \leq C'_{max}$. It follows from (\ref{eq_job}) that
\be
C_{max}  \leq C'_{max} \leq p_{1j} + p_{2j} \leq C^*_{max}.
\ee
That is to say the GAR algorithm will return an optimal solution for this case.
\end{case}
\begin{case}
$J^*\cap D = \emptyset$

Consider the last current schedule during the execution of the GAR algorithm. We denote the corresponding job set and the makespan as $J'$ and $C'_{max}$ respectively.

In this case, we first argue that $J' \cap D = \emptyset$. Suppose that this is not the case, since $J^*\cap D = \emptyset$, the weights of arcs corresponding to the jobs in $J^*$ have not been revised. Hence we have $ (1+\epsilon)\max \left\{\sum_{J_j\in J^*}w^1_j, \sum_{J_j\in J^*}w^2_j\right\}< M$. Moreover, by the assumption $J' \cap D \neq \emptyset$, we have $\max \left\{\sum_{J_j\in J'}w^1_j, \sum_{J_j\in J'}w^2_j\right\} \geq M$. By Theorem \ref{th_minmax}, the solution returned by the ABV algorithm satisfies
\begin{equation*}
M \leq \max \left\{\sum_{J_j\in J'}w^1_j, \sum_{J_j\in J'}w^2_j\right\} \leq (1+\epsilon)\max \left\{\sum_{J_j\in J^*}w^1_j, \sum_{J_j\in J^*}w^2_j\right\} < M,
\end{equation*} which leads to a contradiction.

Notice that each job in the last current schedule satisfies $p_{1j} + p_{2j} < C'_{max}$, since otherwise the algorithm will continue. Therefore, by Theorem \ref{th_o_gs}, the schedule returned by the GS algorithm satisfies
\be
C'_{max} = \max\left\{\sum_{J_j\in J'}p_{1j}, \sum_{J_j\in J'}p_{2j}\right\},\label{eq_go2ar_1}
\ee

Since $J' \cap D = \emptyset$, we know that all jobs $J_j\in J'$ have not been revised. Thus,
it follows from (\ref{eq_max}), (\ref{eq_job}), (\ref{eq_go2ar_1}), Theorem~\ref{th_minmax} and the fact that the schedule returned by the GAR algorithm is the best one among all current schedules, we have
\begin{equation*}\label{eq_om}
\begin{split}
C_{max} \leq C'_{max} &  = \max\left\{\sum_{J_j\in J'}p_{1j}, \sum_{j\in J'}p_{2j}\right\} \\
& = \max\left\{\sum_{J_j\in J'}w^1_{j}, \sum_{J_j\in J'}w^2_{j}\right\}\\
& \leq (1+\epsilon)\max \left\{\sum_{J_j\in J^*}w^1_j, \sum_{J_j\in J^*}w^2_{j}\right\}\\
&\leq(1+\epsilon)C^*_{max}.
\end{split}
\end{equation*}
\end{case}

We have claimed that the algorithms based on UAR($Alg$, $\rho$, $m$) are polynomial-time algorithms. Moreover, notice that the ABV algorithm is a FPTAS \cite{ABV06}, the GS algorithm runs in linear time, and the GAR algorithm implements the ABV algorithms and the GS algorithm at most $n$ times, and we can claim that the GAR algorithm is a FPTAS for $O2|\mathrm{shortest}~\mathrm{path}|C_{max}$.
\qed
\end{proof}

In the following proofs of Theorems \ref{th_RAR}, \ref{th_jjar} and \ref{th_SAR}, we adopt the same notations as in the proof of Theorem \ref{th_GAR}, and also analyze the same two cases. Then we give the proof of the RAR algorithm for $Om|\mathrm{shortest}~\mathrm{path}|C_{max}$.

\begin{proof}[Theorem \ref{th_RAR}]
The argument for the first case is similar to that of Theorem \ref{th_GAR} by noticing that there is at least one job with $\sum_{i=1}^mp_{ij} > \frac{1}{2} C'_{max}$ in both the optimal schedule and one current schedule, and it follows that $C_{max}\leq C'_{max} \leq 2C^*_{max}$.

For the second case, notice that by Theorem \ref{th_o_racsmany}, the makespan of the last current schedule returned by R{\'a}csm{\'a}ny algorithm satisfies $C'_{max} \leq \sum_{J_j\in J'}p_{lj} + \sum_{i = 1}^m p_{ik}$, where $J_k$ is the last completed job and processed on $M_l$. Moreover, each jobs in $J'$ satisfies $\sum_{i=1}^mp_{ij} \leq \frac{1}{2} C'_{max}$, as otherwise the algorithm will continue. By Theorem \ref{th_minmax} and a similar argument as in Theorem \ref{th_GAR}, it is not difficult to show that $C_{max}\leq C'_{max} \leq (2+\epsilon) C^*_{max}$.
\qed
\end{proof}

Before analyzing the performance of $J2|op\leq 2, \mathrm{shortest}~\mathrm{path}|C_{max}$, we first prove Lemma \ref{lemma_J2}.
\begin{proof}[Lemma \ref{lemma_J2}]
Denote $J_{12}$ ($J_{21}$) as the set of jobs with processing order $(1\rightarrow2)$ ($(2\rightarrow1)$) in the original job set. In the schedule returned by Jackson's rule for $J2|op \leq 2|C_{max}$, let $C^J_{max}$ be the makespan, and suppose that the total processing time of $J_{12}$ on $M_1$ is not less than that of $J_{21}$ on $M_2$. By  Jackson's rule, jobs in $J_{21}$ are scheduled after jobs in $J_{12}$ consecutively on $M_1$, therefore no idle occurs on $M_1$. We consider the following cases.
\setcounter{case}{0}
\begin{case}$C^J_{max} = \sum_{J_j \in J}p_{1j}$

It follows from (\ref{eq_max}) that $C^J_{max}\leq \max\{C^{1}_{max}, C^2_{max}\}$.
\end{case}

\begin{case}$C^J_{max} > \sum_{J_j \in J}p_{1j}$

\begin{subcase}{no idle occurs on $M_2$}

In this case, the processes on both machines are consecutive, it is straightforward to show that $C^J_{max} = \sum_{J_j \in J}p_{2j} \leq \max\{C^{1}_{max}, C^2_{max}\}.$
\end{subcase}
\begin{subcase}{idle occurs on $M_2$}

Remember that Jackson's rule first schedules jobs in $J_{12}$ and $J_{21}$ by Johnson's rule respectively, denoted the two schedules as $\sigma_1$ and $\sigma_2$, and then combines the two schedules. Since Johnson's rule returns a permutation schedule, we can denote $\sigma_1 = \{1, \cdots, l\}$ and $J_{12}=\{J_1, \cdots, J_l\}$. Consider the job in $J_{12}$, say $J_k$, which starts the processing on $M_2$ after the last idle on that machine. It is easy to see that $J_k$ starts processing on $M_2$ immediately after its completion on $M_1$. Recall that no idle occurs on $M_1$ by assumption, thus we have $C^J_{max} = \sum^k_{j = 1} p_{1j} + \sum^l_{j = k} p_{2j}$. Notice that $\sigma_1$ is obtained by Johnson's rule. We change all jobs' processing order as $(1\rightarrow2)$ and obtain a schedule by applying Johnson's rule to the revised jobs. Let the makespan of this schedule be $C^1_{max}$. Since this schedule is also obtained by Johnson's rule, we know $J_1, \cdots, J_{k-1}$ are also scheduled before $J_{k}$, whereas $J_{k+1}, \cdots, J_l$ are scheduled after $J_k$ in this schedule. Therefore, we have $C^1_{max} \geq \sum^k_{j = 1} p_{1j} + \sum^l_{j = k} p_{2j}$, and it suffices to show that $C^J_{max}\leq C^1_{max}$.
\end{subcase}
\end{case}

For the case where total processing time of $J_{12}$ on $M_1$ is less than that of $J_{21}$ on $M_2$, an analogous argument also yields $C^J_{max}\leq \max\{C^{1}_{max}, C^2_{max}\}$.
\qed
\end{proof}

Now we can study the performance of the JJAR algorithm for $J2|op\leq 2, \mathrm{shortest}~\mathrm{path}|C_{max}$.
\begin{proof}[Theorem \ref{th_jjar}]
The first case is similar to that of Theorem \ref{th_GAR} by noticing that there is at least one job with $\sum_{i=1}^mp_{ij} > \frac{2}{3} C'_{max}$ in both the optimal schedule and one current schedule, it follows that $C_{max}\leq C'_{max} \leq \frac{3}{2}C^*_{max}$.

For the second case, first by lemma \ref{lemma_J2} we have $C'_{max} \leq \max\{C^1_{max}, C^2_{max}\}$, where $C'_{max}$ is the makespan of the last current schedule, and $C^1_{max}$ ($C^2_{max}$) is the makespan of the schedule obtained by changing the processing order of all jobs of $J'$ to be $(1\rightarrow2)$ $((2\rightarrow1))$ and applying Johnson's rule. Assume that $C'_{max} \leq C^1_{max}$, and denote $J_v$ as the critical job of the schedule with respect to $C^1_{max}$. If $p_{1v}\geq p_{2v}$, we have $C^1_{max} \leq \sum_{J_j\in J'} p_{1j} +  p_{2v}$. Notice that all jobs in the last current schedule satisfy $p_{1j} + p_{2j} \leq \frac{2}{3}C'_{max}$ in the JJAR algorithm, and we have $p_{2v} \leq \frac{1}{3}C'_{max}$. A similar argument as in Theorem \ref{th_GAR} shows that $C_{max}\leq C'_{max} \leq C^1_{max} \leq (\frac{3}{2}+\epsilon)C^*_{max}$. The other situations can be verified by analogous arguments. Thus, the JJAR algorithm is $(\frac{3}{2}+\epsilon)$-approximate.
\qed
\end{proof}

Finally, we give the proof of the SAR algorithm for $Jm|\mathrm{shortest}~\mathrm{path}|C_{max}$.
\begin{proof}[Theorem \ref{th_SAR}]
The first case is analogous, and we can show that $C_{max}\leq \frac{2\alpha\log^2(m\mu)}{\log{\log(m\mu)}}C^*_{max}$.

For the second case, all the jobs in $J'$ satisfy $\sum_{i=1}^m\mu_{ij}p_{ij} \leq \frac{\log{\log(m\mu)}}{2\alpha\log^2(m\mu)} C'_{max}$, as otherwise the algorithm will continue. Combining (\ref{eq:sjar_1}) and Theorem \ref{th_minmax}, by a similar argument as in Theorem \ref{th_GAR}, it is not difficult to show that $C_{max} \leq C'_{max} \leq (1+\epsilon)\frac{2\alpha\log^2(m\mu)}{\log{\log(m\mu)}} C^*_{max}$. Thus, there exists an $O(\log^2(m\mu)/\log{\log(m\mu)})$-approximation algorithm for this problem.
\qed
\end{proof}

\section{The Proof of Theorem \ref{th_SAE} in Section \ref{sec_alg_1+e_o}}\label{app_SAE_proof}
This appendix analyzes the performance of the SAE algorithm.
\begin{proof}[Theorem \ref{th_SAE}]
Remember that we have assumed $N$ is a constant.

Consider the iteration that the subset $J^N$ is exactly the first $N$-th largest jobs of $J^*$, and denote the makespan and the job set returned in this iteration as $C'_{max}$ and $J'$ respectively.

We now argue that the jobs in $J^N$ are also the first $N$-th largest jobs of $J'$. First, the modified ABV algorithm returns a path visiting the arcs corresponding to $J^N$ if such a path exists. Since $J^N$ is exactly the first $N$-th largest jobs of $J^*$, we have $J^*\cap D = \emptyset$ and $J'\cap D = \emptyset$ following the analogous arguments in the proof of the algorithms based on UAR($Alg$, $\rho$, $m$). Therefore $J^N$ is exactly the set of first $N$-th largest jobs of $J'$. Notice that $C_{max}$ is the best one among all current schedules, and we have $C_{max}\leq C'_{max}$, so we only concern about $C'_{max}$ and the schedule returned in that iteration in the subsequent analysis.

Denote $P'_{max} = \max_{i\in \{1, \cdots, m\}}\sum_{J_j\in J'}p_{ij}$. In the SAE algorithm, given $\epsilon > 0$, we can use the modified ABV algorithm to return a path satisfying
\be
P'_{max} \leq \left(1 + \frac{\epsilon}{3}\right) \max_{i\in \{1, \cdots, m\}}\sum_{a_j\in J^*}p_{ij},
\ee
thus from (\ref{eq_max}) we have,
\be
P'_{max}\leq \left(1 + \frac{\epsilon}{3}\right)C^*_{max}. \label{eq_SAE_abv}
\ee

Now we study the schedule returned by the SW algorithm. Recall that jobs are divided into large jobs and small jobs\cite{Sevastianov1998}:
\begin{equation*}
\begin{split}
J'_L =& \{J_j\in J'| p_{ij}\geq \alpha P'_{max}, \mathrm{~for~some~}i, 1\leq i\leq m\},\\
J'_S =& \{J_j\in J'\setminus J'_L \}.
\end{split}
\end{equation*}
Furthermore, the operations of jobs in $J'_S$ are divided into two sets:
\begin{equation*}
\begin{split}
O^1_S =& \{O_{ij}| \alpha^2 P'_{max} < p_{ij} \leq \alpha P'_{max},~~J_j\in J'_S\},\\
O^2_S =& \{O_{ij}| p_{ij}\leq \alpha^2P'_{max},~~J_j\in J'_S\}.
\end{split}
\end{equation*}

The value of $\alpha$ is determined by the inequalities
\be
\left(\frac{\epsilon}{m(3+\epsilon)}\right)^{2^{\frac{m(3+\epsilon)}{\epsilon}}} < \alpha \leq \frac{\epsilon}{m(3+\epsilon)}. \label{eq_SAE_alpha2}
\ee
and
\be
\sum_{O_{ij}\in O^1_S} p_{ij} \leq \frac{\epsilon}{3 + \epsilon} P'_{max}. \label{eq_SAE_alpha1}
\ee
We show that such $\alpha$ exists and can be found in polynomial time. Denote $\alpha_k = \left(\frac{\epsilon}{m(3+\epsilon)}\right)^{2^{k}}$, $k = 0, 1, \cdots, \frac{m(3+\epsilon)}{\epsilon}-1$, and $O^1_S(\alpha_k)$ as the operations of $J'_S$ by setting $\alpha = \alpha_k$. Thus we have $\frac{m(3+\epsilon)}{\epsilon}$ disjoint operation sets $\{O^1_S(\alpha_k)\}$. Notice that the total processing time of all the operations is at most $mP'_{max}$, thus there must be at least one $O^1_S(\alpha_k)$ satisfying (\ref{eq_SAE_alpha1}), then we set $\alpha$ be such $\alpha_k$. Such $\alpha$ can be found in constant time for fixed $m$ and $\epsilon$.

Notice that the number of large jobs satisfy $|J'_L| \leq \frac{m}{\alpha} \leq m\left(\frac{m(3+\epsilon)}{\epsilon}\right)^{2^{\frac{m(3+\epsilon)}{\epsilon}}}  = N$. Consequently, all the large jobs of $J'_L$ belongs to the job set $J^N$, and thus belongs to $J^*$. Moreover, notice that the SW algorithm first find an optimal schedule of $J'_L$ by trying any possible order, denote the makespan of this schedule as $C^L_{max}$. Recall that all jobs in $J^N$ also belong to $J^*$, therefore we have
\be
C^L_{max} \leq C^*_{max}. \label{eq_sae_jn<C*}
\ee

The remaining analysis is based on the SW algorithm \cite{Sevastianov1998}.

Let $M_l$ be the last completed machine, and $t$ be the completion time of the last operation of jobs in $J'_L$ on $M_l$. $T_0$ and $T_t$ are referred to the idle time on $M_l$ during time intervals $[0, t]$ and $[t, C'_{max}]$ respectively. The total processing time of all operations on $M_l$ after time $t$ is denoted by $P_t$. The worst worst-case ratio is shown by considering several cases as follows.

\setcounter{case}{0}
\begin{case}{$P_t = 0$}

It implies that $C'_{max} = t \leq C^L_{max}$. By ($\ref{eq_sae_jn<C*}$), the algorithm returns an optimal schedule.
\end{case}

\begin{case}{$P_t \neq 0$}

Consider the last operation $O_{lj}$ on $M_l$. By assumption, job $J_j$ belongs to $J'_S$. Since $J_j$ cannot be processed on any other idles on $M_l$ after time $t$, at these idles there must be some machine processing $J_j$. By combining (\ref{eq_SAE_abv}) and (\ref{eq_SAE_alpha2}), we have
\be
T_t \leq \sum_{i \neq l}p_{ij} \leq (m-1)\alpha P'_{max} \leq \frac{\epsilon}{3 + \epsilon}P'_{max}\leq \frac{\epsilon}{3}C^*_{max}.\label{eq_sae_t}
\ee
Let $O'$ be the set of $O^2_S$ processed on $M_l$ after $t$. We consider the following subcases.

\begin{subcase}{$|O'|\leq m-1$}

By (\ref{eq_SAE_alpha2}), (\ref{eq_SAE_alpha1}) and (\ref{eq_SAE_abv}), the total processing time on $M_l$ after time $t$ (the jobs are all in $J'_S$) satisfies
\be
P_t \leq \sum_{O_{ij}\in O^1_S} p_{ij} + \sum_{O_{ij}\in O'} p_{ij} \leq \frac{\epsilon}{3 + \epsilon}P'_{max} + (m-1)\alpha^2P'_{max} \leq \frac{2}{3}\epsilon C^*_{max}. \label{eq_SAE_c2}
\ee

Thus, it follows from (\ref{eq_sae_jn<C*}), (\ref{eq_sae_t}) and (\ref{eq_SAE_c2}) that
\be
\begin{split}
C_{max}  \leq C'_{max}&  = t + P_t + T_t \leq C^L_{max} + P_t + T_t\\
&\leq C^*_{max} + \frac{2}{3}\epsilon C^*_{max} + \frac{\epsilon}{3}C^*_{max} = (1 + \epsilon)C^*_{max}.
\end{split}
\ee
\end{subcase}

\begin{subcase}{$|O'|> m-1$}

Notice that there are at most $m/\alpha$ idles (large jobs) before $t$ on machine $M_l$ in the schedule with respect to $C'_{max}$. On each idle, since there are at most $m-1$ jobs being processed on the other machines and $|O'|> m-1$, the idle must smaller than $\alpha^2 P'_{max}$, since otherwise some job will be scheduled on that idle. Therefore we have
\be
T_0 \leq \frac{m}{\alpha}\alpha^2P'_{max} \leq \frac{\epsilon}{3 + \epsilon}P'_{max} \leq \frac{\epsilon}{3}C^*_{max}. \label{eq_SAE_c1}
\ee
Since $C'_{max}$ can be written as $C'_{max} = \sum_{j\in J'}p_{lj} + T_0 + T_t$, it follows from (\ref{eq_SAE_abv}), (\ref{eq_sae_t}) and (\ref{eq_SAE_c1}) that
\be
\begin{split}
C_{max}  \leq C'_{max}&  = \sum_{j\in J'}p_{lj} + T_0 + T_t \leq P'_{max} + T_0 + T_t\\
&\leq (1 + \frac{\epsilon}{3})C^*_{max} + \frac{2}{3}\epsilon C^*_{max} = (1 + \epsilon)C^*_{max}.
\end{split}
\ee

\end{subcase}
\end{case}
In conclusion, the SAE algorithm produces a schedule with makespan at most $(1 + \epsilon)C^*_{max}$.
\qed
\end{proof}
\end{document}